\documentclass[12pt]{article}
\usepackage[margin=25mm]{geometry}                
\usepackage{graphicx}
\usepackage{url}
\usepackage{amssymb,amsmath,amsthm,color}
\usepackage{caption}
\usepackage{subcaption}
\usepackage{microtype}

\newtheorem{theorem}{Theorem}
\newtheorem{lemma}{Lemma}
\newtheorem{proposition}{Proposition}
\theoremstyle{remark}
\newtheorem{remark}{Remark}
\newtheorem{example}{Example}
\newtheorem{question}{Question}
\newtheorem*{corollary}{Corollary}

\DeclareMathOperator{\KS}{\mathrm{C}\mskip 0.7mu}
\DeclareMathOperator*{\E}{\mathbb{E}}
\DeclareMathOperator{\HH}{\mathrm{H}\mskip 0.7mu}

\newcommand{\B}[1]{\mathbb{B}^{#1}}
\newcommand{\Noise}[1]{N_{#1}}
\newcommand{\littleO}{o}
\newcommand{\ind}[1]{\mathbf{1}_{#1}}

\newcommand{\ones}{\mathbf{1}}

\usepackage{hyperref}
\ifxetex
\usepackage{fontspec,xltxtra,xunicode}
\usepackage{polyglossia}
\setdefaultlanguage{english}
\setotherlanguage{russian}
\defaultfontfeatures{Mapping=tex-text}
\setromanfont[Mapping=tex-text]{STIX Two Text}
\usepackage{unicode-math}
\else
\usepackage{mathptmx}
\fi

\usepackage{pgfplots}
\pgfplotsset{compat=1.14}

\def\eps{\varepsilon}
\DeclareMathOperator{\poly}{\mathrm{poly}\mskip 2 mu}
\newcommand{\cnd}{\mskip 2mu |\mskip 2mu}
\let\ge=\geqslant
\let\le=\leqslant

\title{Random noise increases Kolmogorov complexity and Hausdorff dimension\footnote{The work was supported by Russian Academic Excellence Project $5$--$100$, RaCAF ANR-15-CE40-0016-01 grant and RFBR 19-01-00563A grants.}}
\author{
Gleb Posobin\thanks{National Research University Higher School of Economics, \texttt{posobin@gmail.com}. Supported by MK-5379.2018.1 grant. Part of the work done while visiting LIRMM, CNRS, University of Montpellier. Current affiliation: Columbia University, NY, USA. The preparation of the final version was supported by NSF CAREER Award CCF-1844887 and CCF-1563155 grants.}, Alexander Shen\thanks{LIRMM CNRS, Universify of Montpellier, \texttt{alexander.shen@lirmm.fr}. Part of the work done while visiting Toyota Technological University, Chicago.}}
\date{}

\begin{document}
\maketitle

\begin{abstract}
Consider a binary string $x$ of length $n$ whose Kolmogorov complexity is $\alpha n$ for some $\alpha<1$. We want to increase the complexity of $x$ by changing a small fraction of bits in $x$. This is always possible: Buhrman, Fortnow, Newman and Vereshchagin showed  \cite{Buhrman2005} that the increase can be at least $\delta n$ for large $n$ (where $\delta$ is some positive number that depends on $\alpha$ and the allowed fraction of changed bits).

We consider a related question: what happens with the complexity of $x$ when we \emph{randomly} change a small fraction of the bits (changing each bit independently with some probability $\tau$)? We prove that a linear increase in complexity happens with high probability, but this increase is smaller than in the case of arbitrary change. We note that the amount of the increase depends on $x$ (strings of the same complexity could behave differently), and give exact lower and upper bounds for this increase (with $o(n)$ precision).

The proof uses the combinatorial and probabilistic technique that goes back to Ahlswede, G\'acs and K\"orner \cite{Ahlswede1976}. For the reader's convenience (and also because we need a slightly stronger statement) we provide a simplified exposition of this technique, so the paper is self-contained. 

The same technique is used to prove the results about the (effective Hausdorff) dimension of infinite sequences. We show that random change increases the dimension with probability $1$, and provide an optimal lower bound for the dimension of the changed sequence. We also improve a result from~\cite{Greenberg2018} and show that for every sequence $\omega$ of dimension $\alpha$ there exists a strongly $\alpha$-random sequence $\omega'$ such that the Besicovitch distance between $\omega$ and $\omega'$ is $0$.

 \end{abstract}

\section{Introduction}\label{sec:introduction}

The Kolmogorov complexity $\KS(x)$ of a binary string $x$ is defined as the minimal length of a program that generates $x$, assuming that we use an optimal programming language that makes the complexity function minimal up to an $O(1)$ additive term (see~\cite{LiVitanyi2008,ShenUspenskyVereshchagin} for details). There are several versions of Kolmogorov complexity; we consider the original version, called \emph{plain} complexity. In fact, for our considerations the difference between different versions of Kolmogorov complexity does not matter, since they differ only by $O(\log n)$ additive term for $n$-bit strings, but we restrict ourselves to plain complexity for simplicity. 

The complexity of $n$-bit strings is between $0$ and $n$ (we omit $O(1)$ additive terms). Consider a string $x$ of length $n$ that has some intermediate complexity, say $0.5n$. Imagine that we are allowed to change a small fraction of bits in $x$, say, $1\%$ of all bits. Can we decrease the complexity of $x$? Can we increase the complexity of $x$? What happens if we change randomly chosen $1\%$ of bits?

In other words, consider a Hamming ball with center $x$ and radius $0.01n$, i.e., the set of strings that differ from $x$ in at most $0.01n$ positions. What can be said about the minimal complexity of strings in this ball? the maximal complexity of strings in this ball? the typical complexity of strings in this ball?

The answer may depend on $x$: different strings of the same complexity may behave differently if we are interested in the complexities of neighbor strings. For example, if the first half of $x$ is a random string, and the second half contains only zeros, the string $x$ has complexity $0.5n$ and it is easy to decrease its complexity by shifting the boundary between the random part and zero part: to move the boundary to $0.48n$ from $0.5n$ we need to change about $0.01n$ bits, and the complexity becomes close to $0.48n$. On the other hand, if $x$ is a random codeword of an error-correcting code with $2^{0.5n}$ codewords of length $n$ that corrects up to $0.01n$ errors, then $x$ also has complexity $0.5n$, but no change of $0.01n$ (or less) bits can decrease the complexity of~$x$, since~$x$ can be reconstructed from the changed version.

The question about the complexity \emph{decrease} is studied by algorithmic statistics (see~\cite{Vereshchagin2010} or the survey~\cite{Vereshchagin2017}), and the full answer is known. For each $x$ one may consider the function 
\[
d\mapsto \text{(the minimal complexity of strings in the $d$-ball centered at $x$)}.
\]
It starts at $\KS(x)$ (when $d=0$) and then decreases, reaching $0$ at $d=n/2$ (since we can change all bits to zeros or to ones). The algorithmic statistic tells us which functions may appear in this way (see \cite[section 6.2]{Vereshchagin2017} or~\cite[theorem 257]{ShenUspenskyVereshchagin}).\footnote{Note that algorithmic statistics uses a different language. Instead of a string $y$ in the $d$-ball centered at $x$, it speaks about a $d$-ball centered at $y$ and containing $x$. This ball is considered as a statistical model for $x$.}

The question about the complexity \emph{increase} is less studied. It is known that some complexity increase is always guaranteed, as shown in~\cite{Buhrman2005}. The amount of this increase may depend on $x$. If $x$ is a random codeword of an error-correcting code, then the changed version of $x$ contains all the information both about $x$ itself and the places where it was changed. This leads to the maximal increase in complexity. The minimal increase, as shown in~\cite{Buhrman2005}, happens for $x$ that is a random element of the Hamming ball of some radius with center $0^n$. However, the natural question: which functions may appear as 
\[
d\mapsto \text{(the maximal complexity of strings in the $d$-ball centered at $x$)},
\]
remains open.

In our paper we study the \emph{typical} complexity of a string that can be obtained from $x$ by changing a fraction of bits chosen randomly. Let us return to our example and consider again a string $x$ of length $n$ and complexity $0.5n$. Let us change about $1\%$ of bits in $x$, changing each bit independently\footnote{From the probabilistic viewpoint it is more naturally to change all the bits independently with the same probability~$0.01$. Then the number of changed bits is not exactly $0.01n$, but is close to $0.01n$ with high probability.} with probability $0.01$. Does this change increase the complexity of $x$? It depends on the changed bits, but it turns out that \emph{random change increases the complexity of the string with high probability}: we get a string of complexity at least $0.501n$ with probability at least $99\%$, for all large enough $n$ (the result is necessarily asymptotic, since the Kolmogorov complexity function is defined up to $O(1)$ terms).

Of course, the parameters above are chosen only as an example, and the following general statement is true. For some $\tau\in(0,1)$ consider the random noise $N_\tau$ that changes each position in a given $n$-bit string independently with probability $\tau$. 

\begin{theorem}\label{thm:increase}
There exists a strictly positive function $\delta(\alpha,\tau)$ defined for $\alpha,\tau\in (0,1)$ with the following property: for all sufficiently large $n$, for every $\alpha\in(0,1)$, for every $\tau\in(0,1)$, for $\beta=\alpha+\delta(\alpha,\tau)$, and for every $x$ such that $\KS(x)\ge \alpha n$, the probability of the event
\[
 \KS(N_\tau(x))>\beta  n
\]
is at least $1-1/n$. 
\end{theorem}

\begin{remark}
We use the inequality $\KS(x)\ge\alpha n$ (and not an equality $\KS(x)=\alpha n$) to avoid technical problems: the complexity $\KS(x)$ is an integer, and $\alpha n$ may not be an integer.
\end{remark}

\begin{remark}
One may consider only $\tau\le 1/2$ since reversing all bits does not change Kolmogorov complexity (so $\tau$ and $1-\tau$ give the same increase in complexity). For $\tau=1/2$ the variable $N_\tau(x)$ is uniformly distributed in the Boolean cube $\mathbb{B}^n$, so its complexity is close to $n$, and the statement is easy (for arbitrary $\beta<1$). 
\end{remark}

\begin{remark}
 We use $\alpha,\tau$ as parameters while fixing the probability bound as $1-1/n$. As we will see, the choice of this bound is not important: we could use a stronger bound (e.g., $1-1/n^d$ for arbitrary $d$) as well.
\end{remark}

Now a natural question arises: what is the optimal bound in Theorem~\ref{thm:increase}, i.e., the maximal possible value of $\delta(\alpha,\tau)$? In other words, fix $\alpha$ and $\tau$. Theorem~\ref{thm:increase} guarantees that there exists some $\beta>\alpha$ such that every string $x$ of length $n$ (sufficiently large) and complexity at least $\alpha n$ is guaranteed to have complexity at least $\beta n$ after $\tau$-noise $N_\tau$ (with high probability).  \emph{What is the maximal value of $\beta$ for which such a statement is true} (for given $\alpha$ and~$\tau$)?

Before answering this question, we should note that the guaranteed complexity increase depends on $x$: for different strings of the same complexity the typical complexity of $N_\tau(x)$ could be different. Here are the two opposite examples (with minimal and maximal increase, as we will see).

\begin{example}\label{ex:minimal}
Consider some $p\in(0,1)$ and the Bernoulli distribution $B_p$ on the Boolean cube $\mathbb{B}^n$ (bits are independent; every bit equals $1$ with probability $p$). With high probability the complexity of a $B_p$-random string is $o(n)$-close to $n\HH(p)$ (see, e.g.,~\cite[chapter~7]{ShenUspenskyVereshchagin}), where $\HH(p)$ is the Shannon entropy function
\[
\HH(p) = -p \log p - (1-p)\log (1-p).
\] 
After applying $\tau$-noise the distribution $B_p$ is transformed into $B_{N(\tau,p)}$, where 
\[
N(\tau,p)=p(1-\tau)+(1-p)\tau=p+\tau-2p\tau
\]
is the probability to change the bit if we first change it with probability $p$ and then (independently) with probability $\tau$.\footnote{We use the letter $N$ (for ``noise'') both in $N_\tau(x)$ (random change with probability $\tau$, one argument) and in $N(\tau,p)$ (the parameter of the Bernoulli distribution $B_p$ after applying $N_\tau$, no subscript, two arguments).} The complexity of $N_\tau(x)$ is close (with high probability) to $\HH(N(\tau,p))n$ since the $B_p$-random string $x$ and the $\tau$-noise are chosen independently. So in this case we have (with high probability) the complexity increase
\[
 \HH(p) n \to \HH(N(\tau,p))n.
\]
Note that $N(\tau,p)$ is closer to $1/2$ than $p$, and $\HH$ is strictly increasing on $[0,1/2]$, so indeed some increase happens.
\end{example}

\begin{example}\label{ex:maximal}
Now consider an error-correcting code that has $2^{\alpha n}$ codewords and corrects up to $\tau n$ errors (this means that the Hamming distance between codewords is greater than $2\tau n$). Such a code may exist or not depending on the choice of $\alpha$ and $\tau$. The basic result in coding theory, Gilbert's bound, guarantees that such a code exists if $\alpha$ and $\tau$ are not too large. Consider some pair of $\alpha$ and $\tau$ for which such a code exist; moreover, let us assume that it corrects up to $\tau'n$ errors for some $\tau'>\tau$. We assume also that the code itself (the list of codewords) has small complexity, say, $O(\log n)$. This can be achieved by choosing the first (in some ordering) code with required parameters.

Now take a random codeword of this code; most of the codewords have complexity close to $\alpha n$. If we randomly change each bit with probability $\tau$, then with high probability we get at most $\tau' n$ errors, therefore, decoding is possible and the pair $(x,\text{noise})$ can be reconstructed from $N_\tau(x)$, the noisy version of $x$. Then the complexity of $N_\tau(x)$ is close to the complexity of the pair $(x,\text{noise})$, which (due to independence) is close to $\alpha n + \HH(\tau)n$ with high probability. So in this case we have the complexity increase
$[
\alpha n \to (\alpha + \HH(\tau)) n.
$]
\end{example}
\begin{remark}
Note that this increase is the maximal possible not only for  random independent noise but for any change in $x$ that changes a $\tau$-fraction of bits. See below about the difference between random change and arbitrary change.
\end{remark}

Now we formulate the result we promised. It says that the complexity increase observed in Example~\ref{ex:minimal} is the minimal possible: such an increase is guaranteed for every string of given complexity.

\begin{theorem}\label{thm:main}
Let $\alpha=\HH(p)$ for some $p\le 1/2$. Let $\tau$ be an arbitrary number in $(0,1)$. Let $\beta = \HH(N(p,\tau))$. Then for sufficiently large $n$ the following is true: for every string $x$  of length $n$ with $\KS(x)\ge \alpha n$, we have
\[
\Pr[\KS(N_\tau(x))\ge \beta n - o(n)] \ge 1-\frac{1}{n}.
\] 
\end{theorem}

Here $o(n)$ denotes some function such that $o(n)/n\to 0$ as $n\to\infty$. This function does not depend on $\alpha$, $\beta$, and $\tau$. As the proof will show, we may take $o(n)=c\sqrt{n}\log^{3/2}n$ for some $c$.

\begin{figure}[ht]
\begin{center}
	\resizebox{.7\textwidth}{!}{\input{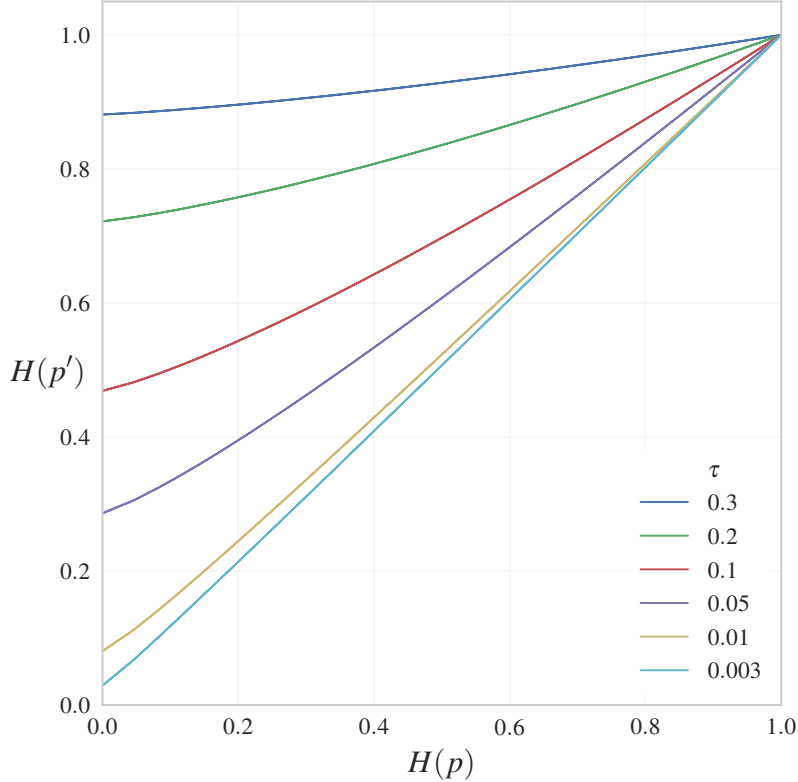}}
\end{center}
\caption{The curves $(\HH(p), \HH(p'))$ where $p'=N(p,\tau)$. Six different values of $\tau$ are shown.}\label{pic:boundary}
\end{figure}

Figure~\ref{pic:boundary} shows the values of $(\alpha,\beta)$ where Theorem~\ref{thm:main} can be applied, for six different values of  $\tau$. Example~\ref{ex:minimal} shows that the value of $\beta$ in this theorem is optimal. 

Theorem~\ref{thm:main} is the main result of the paper. It is proven, as it often happens with results about Kolmogorov complexity, by looking at its combinatorial and probabilistic counterparts. In the next section we explain the scheme of the proof and outline its main ideas. 

Then we explain the details of the proof. It starts with the Shannon information counterpart of our complexity statement that is proven in~\cite{Wyner1973}. In Section~\ref{sec:comb-proof} we derive two different combinatorial counterparts following~\cite{Ahlswede1976}. Finally, in Section~\ref{sec:complexity} we consider the details of the conversion of a combinatorial statement to a complexity one and finish the proof.

In Section~\ref{sec:infinite} we extend our techniques to infinite sequences and compare the results obtained by our tools and the results about arbitrary change of a small fraction of bits from~\cite{Greenberg2018}.

In fact, if we are interested only in \emph{some} complexity increase (Theorem~\ref{thm:increase}), a simple argument (suggested by Fedor Nazarov) that uses Fourier transform is enough. A stronger result (but still not optimal) can be obtained by hypercontractivity techniques. These two arguments are sketched in Appendix A.

In Appendix B, for reader's convenience, we reproduce the proof of the result from~\cite{Ahlswede1976} (about the increase in entropy caused by random noise) used in the proof.

Finally in Appendix C we provide short (and quite standard) proofs of the McDiarmid inequality as a corollary of the Azuma--Hoeffding inequality and of the Azuma--Hoeffding inequality itself.

\section{Proof sketch}
\label{sec:proof-sketch}

\subsection{Three ways to measure the amount of information}

Kolmogorov's  first paper on algorithmic information theory~\cite{Kolmogorov1965} was called ``Three approaches to the Quantitative Definition of Information''. These three approaches can be summarized as follows:
\begin{itemize}
\item (Combinatorial): an element of a set of cardinality $N$ carries $\log N$ bits of information.
\item (Algorithmic): a binary string $x$ carries $\KS(x)$ bits of information, where $\KS(x)$ is the minimal bit length of a program that produces~$x$.
\item (Shannon information theory, or probabilistic approach): a random variable $\xi$ that has $k$ values with probabilities $p_1,\ldots,p_k$, carries $\HH(\xi)$ bits of information, where $\HH(\xi)$ is the Shannon entropy of $\xi$, defined as
\[
\HH(\xi)=p_1\log\frac{1}{p_1} +\ldots + p_k \log\frac{1}{p_k}
\]
\end{itemize}
One cannot compare directly these three quantities since the measured objects are different (sets, strings, random variables). Still these quantities are closely related, and many statements that are true for one of these notions can be reformulated for other two. Several examples of this type are discussed in~\cite[chapters 7 and 10]{ShenUspenskyVereshchagin}, and we use this technique in our proof.

\subsection{Arbitrary change}\label{subsec:arbitrary}

We start by recalling an argument from~\cite{Buhrman2005} for the case when we are allowed to change arbitrary bits (only the number of changed bits is bounded) and want to increase complexity. (A similar reduction will be a part of our argument.)

Fix some parameters $\alpha$ (determining the complexity of the original string), $\tau$ (the maximal fraction of changed bits), and $\beta$ (determining the complexity of the changed string). Let us repeat the complexity statement and give its combinatorial equivalent.
\begin{itemize}
\item (Complexity version) Every string $x$ of length $n$ and complexity at least $\alpha n$ can be changed in at most $\tau n$ positions to obtain a string of complexity at least $\beta n$. 
\item (Combinatorial version) For every subset $B$ of the Boolean cube $\mathbb{B}^n$ of cardinality at most $2^{\beta n}$, its $\tau n$-interior has cardinality at most $2^{\alpha n}$.
\end{itemize}
Here by $d$-interior of a set $X\subset \mathbb{B}^n$ we mean the set of strings $x\in\mathbb{B}^n$ such that the entire ball of radius $d$ centered in $x$ belongs to $X$. In other words, a string $x$ does \emph{not} belong to the $d$-interior of $X$ if $x$ can be changed in at most $d$ positions to get a string outside $X$.
\begin{remark}
The combinatorial statement can be reformulated in a dual way: for every set $A\subset\mathbb{B}^n$ of cardinality greater than $2^{\alpha n}$, its $d$-neighborhood has cardinality greater than~$2^{\beta n}$.
\end{remark}
These two statements (combinatorial and complexity versions) are almost equivalent: one of them implies the other if we allow a small change in $\alpha$ and $\beta$ (in fact, $O(\log n)/n$ change is enough).   Indeed, assume  first that the combinatorial statement is true. Consider the set $B$ of all $n$-bit strings that have complexity less than $\beta n$. Then $\#B < 2^{\beta n}$, so we can apply the combinatorial statement.\footnote{For simplicity we assume that $\alpha n$, $\beta n$, and $\tau n$ are integers. This is not important, since we have $O(\log n)$ term anyway.} It guarantees that the $\tau n$-interior of $B$ (we denote it by $A$) has at most $2^{\alpha n}$ elements. The set $A$ can be enumerated given $n$, $\beta n$ and $\tau n$. Indeed, knowing $n$ and $\beta n$, one can enumerate elements of $B$ (by running in parallel all programs of length less than $\beta n$; note that there are less than $2^{\beta n}$ of them). Knowing also $\tau n$, we may enumerate $A$ (if a ball is contained in $B$ entirely, this will become known at some stage of the enumeration of $B$). Then every element of $A$ can be encoded by its ordinal number in this enumeration. This guarantees that the complexity of all elements of $A$ does not exceed $\alpha n+O(\log n)$ (the additional $O(\log n)$ bits are needed to encode $n$, $\beta n$, and $\tau n$). Therefore, if some $x$ has complexity greater that $\alpha n +O(\log n)$, it is not in $A$, i.e., $x$ can be changed in at most $\tau n$ positions to get a string outside $B$. By the definition of $B$, this changed string has complexity at least $\beta n$, as required. The term $O(\log n)$ can be absorbed by a small change in $\alpha$.

Let us explain also (though this direction is not needed for our purpose) why the complexity statement implies the combinatorial one. Assume that there are some sets $B$ that violate the combinatorial statement, i.e., contain at most $2^{\beta n}$ strings but have $\tau n$-interior of size greater than $2^{\alpha n}$. Such a set can be found by exhaustive search, and the first set $B$ that appears during the search has complexity $O(\log n)$.  Its elements, therefore, have complexity $\beta n +O(\log n)$: to specify an element, we need to specify $B$ and the ordinal number of the element in $B$. From this we conclude, using the complexity statement (and changing $\beta$ slightly) that all elements of the $\tau n$-interior of $B$ have complexity at most $\alpha n$. Therefore, there are at most $O(2^{\alpha n})$ of them, and the size of the interior is bounded by $2^{\alpha n}$ (again up to a small change in $\alpha$).

Now we return to the result from~\cite{Buhrman2005}. Let $x$ be a string of length $n$ and complexity at least $\alpha n+O(\log n)$, where $\alpha = \HH(p)$ for some $p\le 1/2$. Let $\tau$ be a real such that $p+\tau\le 1/2$, and $\beta = \HH(p+\tau)$. Then $x$ can be changed in at most $\tau n$ positions to get a string of complexity at least $\beta n$. As we have seen, this statement from~\cite{Buhrman2005} is a corollary of the following combinatorial result.

\begin{proposition}\label{prop:harper}
Let $p\le 1/2$ be some number and let $\alpha=H(p)$. Let $\tau$ be some positive number so that $p+\tau\le 1/2$, and let $\beta=\HH(p+\tau)$. Let $B$ be an arbitrary subset of $\,\mathbb{B}^n$ of size at most $2^{\beta n}$. Let $A$ be a subset of $\,\mathbb{B}^n$, and for every $x\in A$ the Hamming ball of radius $\tau n$ with center $x$ is contained in $B$. Then the cardinality of $A$ does not exceed $\poly(n)2^{\alpha n}$. 
\end{proposition}

This proposition is a direct consequence of Harper's theorem  (see, e.g.,~\cite{Frankl1981}) that says that for a subset of $\mathbb{B}^n$ of a given size, \emph{its $d$-interior} (for some fixed $d$) \emph{is maximal when the subset is a Hamming ball} (formally speaking, is between two Hamming balls of sizes $k$ and $k+1$ for some $k$). Or, in dual terms, Harper's theorem says that the \emph{$d$-neighborhood of a set of a given size is minimal if this set is a Hamming ball}. The relation between $2^{\alpha n}$ and $2^{\beta n}$ in the proposition is just the relation between the sizes of balls of radii $pn$ and $(p+\tau)n$ (if we ignore factors that are polynomial in $n$). Note that $p+\tau\le 1/2$ is needed since otherwise the radius exceeds $n/2$ and then the log-size of the ball is close to $n$ and not to $\HH(p+\tau)n$. The $\poly(n)$ factor is needed due to the polynomial factor in the estimate for the ball size in terms of Shannon entropy (the ball of radius $\gamma n$ has size $\poly(n)2^{\HH(\gamma)n}$).

We do not go into details here (and do not reproduce the proof of Harper's theorem) since we need this result only to motivate the corresponding relation between combinatorial and complexity statements for the case of a random noise we are interested in.

\subsection{Random noise: four versions}

For the random noise case we need a more complicated argument. First, we need to consider also the probabilistic version of the statement (in addition to the complexity and combinatorial versions). Second, we need \emph{two} combinatorial versions (strong and weak). Fix some $\alpha$, $\beta$ and $\tau$. Here are the four versions we are speaking about; all four statements are equivalent (are true for the same parameters $\alpha$, $\tau$, and $\beta$, up to $o(1)$-changes in the parameters):

\begin{itemize}
\item (Shannon information version, \cite{Wyner1973}) For every random variable $P$ with values in $\mathbb{B}^n$ such that $\HH(P)\ge\alpha n$, the variable $N_\tau(P)$ that is obtained from $P$ by applying independent noise changing each bit with probability $\tau$, has entropy $\HH(N_\tau(P))\ge \beta n$.

\item (Complexity version) For every string $x$ of length $n$ and complexity $\KS(x)\ge \alpha n$, the probability of the event ``$\KS(N_\tau(x))\ge \beta n$'' is at least $1-1/n$. (Again, $N_\tau$ is random noise that independently changes each bit with probability $\tau$, but now it is applied to the string $x$ and not to a random variable) 

\item (Strong combinatorial version) For every set $B\subset\mathbb{B}^n$ of size at most $2^{\beta n}$ the set $A$ of all strings $x$ such that $\Pr[N_\tau(x)\in B]\ge 1/n$ has size $\#A\le 2^{\alpha n}$.

\item (Weak combinatorial version) For every set $B\subset\mathbb{B}^n$ of size at most $2^{\beta n}$ the set $A$ of all strings $x$ such that $\Pr[N_\tau(x)\in B]\ge 1-1/n$ has size $\#A\le 2^{\alpha n}$.
\end{itemize}

The difference between weak and strong combinatorial versions is due to the different thresholds for the probability. In the weak version the set $A$ contains only strings that get into $B$ after the noise \emph{almost surely} (with probability at least $1-1/n$). In the strong version the set $A$ is bigger and includes all strings that get into $B$ \emph{with non-negligible probability} (at least $1/n$), so the upper bound for $\#A$ becomes a stronger statement.

\begin{remark}
In the case of arbitrary changes (the result from~\cite{Buhrman2005}) we consider the $\tau n$-interior of $B$, the set of points that remain in $B$ after arbitrary change in (at most) $\tau n$ positions. If a point is \emph{not} in the interior, it can be moved outside $B$ by changing at most $\tau n$ bits. Now we consider (in the strong version) the set of points that get into $B$ with probability at least $1/n$. If a point is \emph{not} in this set, the random $\tau$-noise will move it outside $B$ almost surely (with probability at least $1-1/n$). Again the complexity and (strong) combinatorial versions are equivalent up to $o(1)$ changes in parameters, for the same reasons.
\end{remark}

This explains why we are interested in the strong combinatorial statement. The weak one is used as an intermediate step in the chain of arguments. This chain goes as follows:
\begin{itemize}

\item First the Shannon entropy statement is proven using tools from information theory (one-letter characterization and inequalities for Shannon entropy); this was done in~\cite{Wyner1973}.

\item Then we derive the  weak combinatorial statement from the entropy statement using a simple coding argument from~\cite{Ahlswede1976}.

\item Then we show that weak combinatorial statement implies the strong one, using a tool that is called the ``blowing-up lemma'' in~\cite{Ahlswede1976} (now it is more popular under the name of ``concentration inequalities'').

\item Finally, we note that the strong combinatorial statement implies the complexity statement (using the argument sketched above).

\end{itemize}

\subsection{Tools used in the proof}

Let us give a brief description of the tools used in these arguments.

To prove the Shannon entropy statement, following~\cite{Wyner1973}, fix some $\tau$. Consider the set $S_n$ of all pairs $(\HH(P),\HH(N_\tau(P)))$ for all random variables with values in $\mathbb{B}^n$. For each $n$ we get a subset of the square $[0,n]\times [0,n]$. For $n=1$ it is a curve made of all points $(\HH(p),\HH(N(p,\tau)))$ (shown in Figure~\ref{pic:boundary} for six different values of $\tau$). We start by showing that this curve is convex (performing some computation with power series). Then we show, using the convexity of the curve and some inequalities for entropies, that for every $n$ the set $S_n$ is above the same curve (scaled by factor $n$), and this is the entropy statement we need. See Appendix B for details.

To derive the weak combinatorial statement from the entropy statement, we use a coding argument. Assume that two sets $A$ and $B$ are given, and for every point $x\in A$ the random point $N_\tau(x)$ belongs to $B$ with probability at least $1-1/n$. Consider a random variable $U_A$ that is uniformly distributed in $A$. Then $\HH(U_A)=\log\# A$, and if $\#A\ge 2^{\alpha n}$, then $\HH(U_A)\ge\alpha n$ and $\HH(N_\tau(U_A))\ge \beta n$ (assuming the entropy statement is true for given $\alpha$, $\beta$, and $\tau$). On the other hand, the variable $N_\tau(U_A)$ can be encoded as follows: 
\begin{itemize}
\item one bit (flag) says whether $N_\tau(U_A)$ is in $B$ or not;
\item if yes, then $\log \#B$ bits are used to encode an element of $B$;
\item otherwise $n$ bits are used to encode the value of $N_\tau(U_A)$ (trivial encoding).
\end{itemize}
The average length of this code for $N_{\tau}(U_A))$ does not exceed
\[
1+ \left(1-\frac{1}{n}\right)\log\#B + \frac{1}{n}\cdot n \le \log \#B +O(1).
\]
(Note that if the second case has probability less than $1/n$, the average length is even smaller.) The entropy of a random variable $N_\tau(U_A)$ does not exceed the average length of the code. So we get $\beta n \le \HH(N_\tau(U_A)) \le \log \#B+O(1)$ and $\log \#B \ge \beta n - O(1)$, assuming that $\log \#A \ge \alpha n$.

The next step is to derive the strong combinatorial version from the weak one. Assume that two sets $A, B\subset \mathbb{B}^n$ are given, and for each $x\in A$ the probability of the event $N_\tau(x)\in B$ is at least $1/n$. For some $d$ consider the set $B_d$, the $d$-neighborhood of $B$. We will prove (using the concentration inequalities) that for some $d=o(n)$ the probability of the event $N_\tau(x)\in B_d$ is at least $1-1/n$ (for each $x\in A$). So one can apply the weak combinatorial statement to $B_d$ and get a lower bound for $\#B_d$. On the other hand, there is a simple upper bound: $\#B_d \le \# B\times(\text{the size of $d$-ball})$; combining them, we get the required bound for $\#B$. See Section~\ref{sec:comb-proof} for details.

\begin{remark}
One may also note (though it is not needed for our purposes) that the entropy statement is an easy corollary of the complexity statement, and therefore all four are equivalent up to small changes in parameters. This can be proven in a standard way. Consider $N$ independent copies of random variable $P$ and independently apply noise to all of them. Then we write the inequality for the typical values of the complexities; in most cases they are close to the corresponding entropies (up to $o(N)$ error).  Therefore, we get the inequality for entropies with $o(N)$ precision (for $N$ copies) and with $o(1)$ precision for one copy (the entropies are divided by $N$). As $N\to\infty$, the additional term $o(1)$ disappears and we get an exact inequality for entropies.
\end{remark}

\section{Combinatorial version}\label{sec:comb-proof}

Recall the entropy bound from~\cite{Wyner1973} discussed above  (we reproduce its proof in Appendix~B):

\begin{proposition}\label{prop:noise-entropy}
Let $P$ be an arbitrary random variable with values in $\mathbb{B}^n$, and let $P'=N_\tau(P)$ be its noisy version obtained by applying $N_\tau$ independently to each bit in $P$. Choose $p\le 1/2$ in such a way that $\HH(P)=n\HH(p)$. Then consider $q=N(p,\tau)$, the probability to get $1$ if we apply $N_\tau$ to a variable that equals $1$ with probability $p$. Then $\HH(P')\ge n\HH(q)$.
\end{proposition}

In this section we use this entropy bound to prove the combinatorial bounds. We start with the weak one and then amplify it to get the strong one, as discussed in Section~\ref{sec:proof-sketch}. First, let us formulate explicitly the weak bound that is derived from Proposition~\ref{prop:noise-entropy} using the argument of Section~\ref{sec:proof-sketch}.

\begin{proposition}\label{prop:weak-combinatorial}
Let $\alpha=\HH(p)$ and $\beta=\HH(N(p,\tau))$. Let $A,B\subset \mathbb{B}^n$ and for every $x\in A$ the probability of the event ``$N_\tau(x)\in B$'' is at least $1-1/n$. If $\log\# A \ge \alpha n$, then $\log \#B\ge \beta n -O(1)$.
\end{proposition}
In fact, this ``$O(1)$'' is just $2$, but we do not want to be too specific here.
\medskip

Now we need to extend the bound to the case when the probability of the event $N_\tau(x)\in B$ is at least $1/n$. We already discussed how this is done. Consider for some $d$ (depending on~$n$) the Hamming $d$-neighborhood  $B_d$ of $B$. We need to show that
\[
\Pr[N_\tau(x)\in B]\ge \frac{1}{n} \Rightarrow
 \Pr[N_\tau(x)\in B_d]\ge 1-\frac{1}{n}.
\]
for every $x\in \mathbb{B}^n$ (for a suitable $d$). In fact, $x$ does not matter here: we may assume that $x=0\!\ldots\! 0$ (flipping bits in $x$ and $B$ simultaneously). In other terms, we use the following property of the Bernoulli distribution with parameter $\tau$: \emph{if some set $B$ has probability not too small according to this distribution, then its neighborhood $B_d$ has probability close to $1$}. We need this property for $d=o(n)$, see below about the exact value of $d$.

Such a statement is called a \emph{blowing-up lemma} in~\cite{Ahlswede1976}. There are several (and quite different) ways to prove statements of this type. The original proof in \cite{Ahlswede1976} used a result of Margulis from~\cite{Margulis1974} that says that the (Bernoulli) measure of a boundary of an arbitrary set $U\subset\mathbb{B}^{n}$ is not too small compared to the measure of a boundary of a ball of the same size. Iterating this statement (a neighborhood is obtained by adding boundary layer several times), we get the lower bound for the measure of the neighborhood. Another proof was suggested by Marton~\cite{Marton1986}; it is based on the information-theoretical considerations that involve transportation cost inequalities for bounding measure concentration. In this paper we provide a simple proof that uses the McDiarmid inequality~\cite{McDiarmid1989}, a simple consequence of the Azuma--Hoeffding inequality~\cite{Hoeffding1963}.
This proof works for $d=O(\sqrt{n \log n})$.

Let us state the blowing-up lemma in a slightly more general version than we need. Let $X_1,\ldots, X_n$ be (finite) probability spaces. Consider the space $X=X_1\times\ldots\times X_n$ with the product measure $\mu$ (so the coordinates are independent) and Hamming distance $d$ (the number of coordinates that differ). In our case $X=\mathbb{B}^n$ and $\mu$ is the Bernoulli measure with parameter $\tau$. The blowing-up lemma says, informally speaking, that if a set is not too small, then its neighborhood has small complement (the size is measured by~$\mu$). It can be reformulated in a more symmetric way: \emph{if two sets are not too small, then the distance between them is rather small}. (Then this symmetric statement is applied to the original set and the complement of its neighborhood.) Here is the symmetric statement.

\begin{proposition}[Blowing-up lemma, symmetric version]
Let $B,B'$ be two subsets of $X=X_1\times\ldots\times X_n$ with the product measure $\mu$. Then 
\[
d(B,B')\le \sqrt{(n/2)\ln(1/\mu(B))}+\sqrt{(n/2)\ln(1/\mu(B'))}.
\]
\end{proposition}

To prove the blowing-up lemma, we use the McDiarmid concentration inequality:  
   
\begin{proposition}[McDiarmid's inequality, \cite{McDiarmid1989}]\label{prop:McDiarmid}
	Consider a function $f\colon X_1\times\ldots\times X_n \to \mathbb{R}$. Assume that changing the $i$-th coordinate changes the value of $f$ at most  by some $c_i$:
	\[
	 |f(x) - f(x')| \le c_i,
	\]
if $x$ and $x'$ coincide everywhere except for the $i$th coordinate. Then
	\[
		\Pr[f-\E f \ge z] \le \exp\left(-\frac{2z^2}{\sum_{i=1}^{n}c_i^2}\right) 
	\]
for arbitrary $z\ge0$.
\end{proposition}
Here the probability and expectation are considered with respect to the product distribution $\mu$ (the same as in the blowing-up lemma, see above). This inequality shows that $f$ cannot be much larger than its average on a big set. Applying this inequality to $-f$, we get the same bound for the points where the function is less than its average by $z$ or more.

For the reader's convenience, we reproduce the proof of the McDiarmid inequality in Appendix~C. 

Now let us show why it implies the blowing-up lemma (in the symmetric version).
\begin{proof}[Proof of the blowing-up lemma]
Let  $f(x) = d(x, B)$ be the distance between $x$ and $B$, i.e., the minimal number of coordinates that one has to change in $x$ to get into $B$. This function satisfies the bounded differences property with $c_i = 1$, so we can apply the McDiarmid inequality to it. Let $m$ be the expectation of $f$. The function $f$ equals zero for arguments in $B$ and therefore is below its expectation at least by $m$ (everywhere in $B$), so
	\[
	\mu(B) \le \exp\left(-\frac{2m^2}{n}\right), \quad \text{or} \quad
   m\le \sqrt{(n/2)\ln(1/\mu(B))}
	\]
On the other hand, the function $f$ is at least $d(B,B')$ for arguments in $B'$, so it exceeds its expectation at least by $d(B,B')-m$ (everywhere in $B'$), therefore the McDiarmid inequality gives
	\[
d(B,B')-m\le \sqrt{(n/2)\ln(1/\mu(B'))},
	\]
and it remains to combine the last two inequalities.
\end{proof}

Here is the special case of the blowing-up lemma we need:
\begin{corollary}
If $\mu$ is a distribution on $\mathbb{B}^n$ with independent coordinates, and $B\subset\mathbb{B}^n$ has measure $\mu(B)\ge1/n$, then for $d=O(\sqrt{n\log n})$ we have $\mu(B_d)\ge 1-1/n$.
\end{corollary}
Indeed, we may apply the blowing-up lemma to $B$ and $B'$, where $B'$ is a complement of $B_d$. If both $B$ and $B'$ have measures at least $1/n$, we get a contradiction for $d\ge2\sqrt{(n/2)\ln n}$ (the distance between $B$ and the complement of its neighborhood $B_d$ exceeds $d$).

\begin{remark}\label{rem:smaller}
In the same way we get a similar result for probabilities $1/n^c$ and $1-1/n^c$ for arbitrary constant $c$  (only the constant factor in $O(\sqrt{n \log n})$ will be different).
\end{remark}

Now we are ready to prove the strong combinatorial version:
\begin{proposition}\label{prop:strong-combinatorial}
Let $\alpha=\HH(p)$ and $\beta=\HH(N(p,\tau))$. Let $A,B\subset \mathbb{B}^n$ and for every $x\in A$ the probability of the event ``$N_\tau(x)\in B$'' is at least $1/n$. If $\log\# A \ge \alpha n$, then $\log \#B\ge \beta n -O(\sqrt{n}\log^{3/2}n)$.
\end{proposition}

\begin{proof}
As we have seen, the weak combinatorial version (Proposition~\ref{prop:weak-combinatorial}) can be applied to the neighborhood $B_d$ for $d=O(\sqrt{n\log n})$. The size of $B_d$ can be bounded by the size of $B$ multiplied by the size of a Hamming ball of radius $d$. The latter is $\poly(n)2^{n\HH(d/n)}$. Combining the inequalities, we get
\[
\log \# B \ge \log \#B_d - n\HH(d/n)-O(\log n) \ge \beta n - n\HH(d/n)-O(\log n).
\]
For small $p$ we have
\[
\HH(p)=p\log\frac{1}{p}+(1-p)\log \frac{1}{1-p}=p\log\frac{1}{p}+p + o(p)=O\left(p\log\frac{1}{p}\right).
\]
We have $p=d/n=O(\sqrt{\log n/n})$, so 
\[
n\HH(d/n)=nO(\sqrt{\log n/n}\log n)=O(\sqrt{n}\log^{3/2}n),
\]
 as promised.
\end{proof}

\section{Complexity statement}\label{sec:complexity}

Now we combine all pieces and prove Theorem~\ref{thm:main}. It states:
\begin{quote}
Let $\alpha=\HH(p)$ for some $p\le 1/2$. Let $\tau$ be an arbitrary number in $(0,1)$. Let $\beta = \HH(N(p,\tau))$. Then for sufficiently large $n$ the following is true: for every string $x$ of length $n$ with $\KS(x)\ge \alpha n$, we have
$$
\Pr[\KS(N_\tau(x))\ge \beta n - o(n)] \ge 1-\frac{1}{n}.
$$ 
\end{quote}
Here $o(n)$ is actually $O(\sqrt{n}\log^{3/2}n)$.
\smallskip

We already have all the necessary tools for the proof, but some adjustments are needed. We already know how to convert a combinatorial statement into a complexity one. For that we consider the set $B$ of all strings in $\mathbb{B}^n$ that have complexity less than $\beta n - c\sqrt{n}\log^{3/2}n$ for some $c$ (to be chosen later). Then we consider the set $A$ of all $x$ such that $\Pr[N_\tau(x) \in B]\ge 1/n$. The combinatorial statement (strong version, Proposition~\ref{prop:strong-combinatorial}) guarantees that $\# A \le 2^{\alpha n}$. We would like to conclude that all elements of $A$ have complexity only slightly exceeding $\alpha n$. (Then we have to deal with this excess, see below.)  For that we need an algorithm that enumerates $A$. First, we need to enumerate $B$, and for that it is enough to know $n$ and the complexity bound for elements of $B$. But now (unlike the case of arbitrary change where we need to know only the maximal number of allowed changes) we need to compute the probability $\Pr[N_\tau(x)\in B]$, and the value of $\tau$ may not be computable, and an infinite amount of information is needed to specify $\tau$. How can we overcome this difficulty? 

Note that it is enough to enumerate some set that contains $A$ but has only slightly larger size. Consider some rational $\tau'$ that is close to $\tau$ and the set 
\[
A' = \{ x\colon \Pr[N_{\tau'}(x)\in B]> 1/2n\}
\]
The combinatorial statement remains true (as we noted in Remark~\ref{rem:smaller}, even $1/n^c$ would be OK, not only $1/2n$), so we may still assume that $\#A'\le 2^{\alpha n}$. We want $A'\supset A$. This will be guaranteed if the difference between $\Pr[N_\tau(x)\in B]$ and $\Pr[N_{\tau'}(x)\in B]$ is less than $1/2n$. To use the coupling argument, let us assume that $N_{\tau}(x)$ and $N_{\tau'}(x)$ are defined on the same space: to decide whether the noise changes $i$th bit, we generate a fresh uniformly random real in $[0,1]$ and compare it with thresholds $\tau$ and $\tau'$. This comparison gives different results if this random real falls into the gap between $\tau$ and $\tau'$. Using the union bound for all bits, we conclude that $\Pr[N_\tau(x)\ne N_{\tau'}(x)]$ in this setting is bounded by $n|\tau'-\tau|$. Therefore, if the approximation error $|\tau'-\tau|$ is less than $1/2n^2$, we get the desired result, and to specify $\tau'$ that approximates $\tau$ with this precision we need only $O(\log n)$ bits. This gives us the following statement:
\begin{quote}
for every string $x$ of length $n$ with $\KS(x)\ge \alpha n+O(\log n)$, we have
\[
\Pr[\KS(N_\tau(x))\ge \beta n - o(n)] \ge 1-\frac{1}{n}.
\] 
\end{quote}
The only difference with the statement of Theorem~\ref{thm:main} is that we have a stronger requirement $\KS(x)\ge \alpha n +O(\log n)$ instead of $\KS(x)\ge \alpha n$. To compensate for this, we need to decrease $\alpha$ a bit and apply the statement we have proven to $\alpha'=\alpha - O(\log n/n)$. Then the corresponding value of $\beta$ also should be changed, to get a point $(\alpha',\beta')$ on the curve (Figure~\ref{pic:boundary}) on the left of the original point $(\alpha,\beta)$. Note that the slope of the curve is bounded by $1$ (it is the case at the right end where the curve reaches $(1,1)$, since the curve is above the diagonal $\alpha=\beta$, and the slope increases with $\alpha$ due to convexity). Therefore, the difference between $\beta$ and $\beta'$ is also $O(\log n/n)$ and is absorbed by the bigger term $O(\sqrt{n}\log^{3/2}n)$. 

Theorem~\ref{thm:main} is proven.
\medskip

In the next section we apply our technique to get some related results about infinite bit sequences and their effective Hausdorff dimension.  We finish the part about finite strings with the following natural question.
\begin{question}
 Fix some $x$ and apply random noise $N_\tau$. The complexity of $N_\tau(x)$ becomes a random variable. What is the distribution of this variable? The blowing-up lemma implies that it is concentrated around some value. Indeed, if we look at strings below $1\%$-quantile and above $99\%$-quantile, the blowing-up lemma guarantees that the Hamming distance between these two sets is at most $O(\sqrt{n})$, and therefore the thresholds for Kolmogorov complexity differ at most by $O(\sqrt{n}\log n)$ (recall that for two strings of length $n$ that differ in $i$ positions, their complexities differ at most by $O(i\log n)$, since it is enough to add information about $i$ positions and each position can be encoded by $\log n$ bits).
 
So with high probability the complexity of $N_\tau(x)$ is concentrated around some value (defined up to $O(\sqrt{n}\log n)$ precision). For each $\tau$ we get some number (expected complexity, with guaranteed concentration) that depends not only on $n$ and $\KS(x)$, but on some more specific properties of $x$. What are these properties?  Among the properties of this type there exists a Vitanyi--Vereshchagin profile curve for balls, the \emph{minimal} complexity in the neighborhood as function of the radius (see~\cite[section 14.4]{ShenUspenskyVereshchagin}); is it somehow related?

As we have mentioned, this question is open also for \emph{maximal} complexity in $d$-balls around~$x$, not only for \emph{typical} complexity after $\tau$-noise.
\end{question}

\section{Infinite sequences and Hausdorff dimension}\label{sec:infinite}

Let $X=x_1x_2x_3\ldots$ be an infinite bit sequence. The effective Hausdorff dimension of $X$ is defined as
\[
\liminf_{n\to\infty} \frac{\KS(x_1\ldots x_n)}{n}.
\]
A natural question arises: \emph{what happens with the Hausdorff dimension of a sequence when each its bit is independently changed with some probability $\tau$}? The following result states that the dimension increases with probability~$1$ (assuming the dimension was less than $1$, of course), and the guaranteed increase follows the same curve as for finite sequences.

\begin{theorem}\label{thm:dimension-increase}
Let $p,\tau\in (0,1/2)$ be some reals, $\alpha=\HH(p)$ and $\beta=\HH(N(p,\tau))$. Let $X$ be an infinite sequence that has effective Hausdorff dimension at least $\alpha$. Then the effective Hausdorff dimension of the sequence $N_\tau(X)$ that is obtained from $X$ by applying random $\tau$-noise independently to each position, is at least $\beta$ with probability $1$.
\end{theorem}

\begin{proof}
It is enough to show, for every $\beta'<\beta$, that the dimension of $N_\tau(X)$ is at least $\beta'$ with probability $1$.  Consider $\alpha'<\alpha$ so that the pair $(\alpha',\beta')$ lies on the boundary curve. By definition of the effective Hausdorff dimension, we know that $\KS(x_1\ldots x_n)>\alpha' n$ for all sufficiently large $n$. Then Theorem~\ref{thm:main} can be applied to $\alpha'$ and $\beta'$. It guarantees that with probability at least  $1-1/n$ the changed string has complexity at least $\beta' n - o(n)$. Moreover, as we have said, the same is true with probability at least $1-1/n^2$. This improvement is important for us:  the series $\sum 1/n^2$ converges, so the Borel--Cantelli lemma says that with probability $1$ only finitely many prefixes have complexity less than $\beta'n-o(n)$, therefore the dimension of $N_\tau(X)$ is at least $\beta'$ with probability $1$.
\end{proof}

In the next result we randomly change bits with probabilities depending on the bit position. The probability of change in the $n$th position converges to $0$ as $n\to\infty$. This guarantees that with probability $1$ we get a sequence that is Besicovitch-close to a given one. Recall that the Besicovitch distance between two bit sequences $X=x_1x_2\ldots$ and $Y=y_1y_2\ldots$ is defined as 
$$
\limsup_{n\to\infty}\frac{d(x_1\ldots x_n, y_1\ldots y_n)}{n}
$$
(where $d$ stands for the Hamming distance). So $d(X,Y)=0$ means that the fraction of different bits in the $n$-bit prefixes of two sequences converges to $0$ as $n\to\infty$. The strong law of large numbers implies that if we start with some sequence $X$ and change $i$th bit independently with probability $\tau_i$ with $\lim_n \tau_n=0$, we get (with probability $1$) the sequence $X'$ such that the Besicovitch distance between $X$ and $X'$ is $0$. This allows us to prove the following result using a probabilistic argument. 

\begin{theorem}\label{thm:strong}
Let $X=x_1x_2\ldots$ be a bit sequence whose effective Hausdorff dimension is at least $\gamma$ for some $\gamma<1$. Let $\delta_n$ be a sequence of positive reals such that $\lim_n \delta_n=0$. Then there exists a sequence $X'=x'_1x'_2\ldots$ such that:
\begin{itemize}
\item the Besicovitch distance between $X$ and $X'$ is~$0$;
\item $\KS(x'_1\ldots x'_n)$ is at least $n(\gamma+\delta_n)$ for all sufficiently large~$n$.
\end{itemize}
\end{theorem}

\begin{proof}
For this result we use some decreasing sequence $\tau_i\to 0$ and change $i$th bit with probability $\tau_i$. Since $\tau_i\to 0$, with probability $1$ the changed sequence is Besicovitch-equivalent (distance~$0$) to the original one. It remains to prove that the probability of the last claim (the lower bound for complexities) is also $1$ for the changed sequence, if we choose $\tau_i\to 0$ in a suitable way.

To use different $\tau_i$ for different $i$, we have to look again at our arguments. We start with Proposition~\ref{prop:noise-entropy}: the proof (see Appendix B) remains valid if each bit is changed independently with probability $\tau_i\ge \tau$ depending on the bit's position ($i$). Indeed, for every $\tau'\ge\tau$ the corresponding $\tau'$-curve is above the $\tau$-curve, so the pairs of entropies (original bit, bit with noise) are above the $\tau$-curve and we may apply the same convexity argument.

The derivation of the combinatorial statement (first the weak one, then the strong one) also remains unchanged. The proof of the weak version does not mention the exact nature of the noise at all; in the strong version we use only that different bits are independent (to apply the McDiarmid inequality and the blowing-up lemma). The only problem arises when we derive the complexity version from the combinatorial one. In our argument we need to know $\tau$ (or some approximation for $\tau$) to enumerate $A$. If for each bit we have its own value of $\tau$, even one bit to specify this value is too much for us.

To overcome this difficulty, let us agree that we start with $\tau_i=1/2$, then change them to $1/4$ at some point, then to $1/8$ etc. If for $n$th bit we use $\tau_n=2^{-m}$, then to specify all the $\tau_i$ for $i\le n$ we need to specify $O(m\log n)$ bits (each moment of change requires $O(\log n)$ bits). For $\tau=2^{-m}$ we choose a pair $(\alpha,\beta)$ on the $\tau$-curve such that $\alpha<\gamma<\beta$. To decide when we can start using this value of $\tau$, we wait until $\KS(x_1\ldots x_n)>\alpha n + O(m\log n)$ becomes true and stays true forever, and also $\gamma+\delta_n<\beta-O(\sqrt{n}\log^{3/2}n)$ becomes and stays true. Note that $m$ is fixed when we decide when to start using $\tau=2^{-m}$, so such an $n$ can be found. In this way we guarantee that the probability that $x'_1\ldots x'_n$ will have complexity more than $(\gamma+\delta_n)$  is at least $1-1/n^2$ (we need a converging series, so we use the bound with $n^2$), and it remains to apply the Borel--Cantelli lemma.
\end{proof}

Theorem~\ref{thm:strong} implies that for every $X$ that has effective Hausdorff dimension $\alpha$ there exist a Besicovitch equivalent $X'$ that is $\alpha$-random (due to the complexity criterion for $\alpha$-randomness, see~\cite{Greenberg2018}), and we get the result of~\cite[Theorem 2.5]{Greenberg2018} as a corollary. Moreover, we can get this result in a stronger version than in~\cite{Greenberg2018}, since for slow converging sequence $\delta_n$, for example, $\delta_n=1/\log n$, we get \emph{strong} $\alpha$-randomness instead of \emph{weak} $\alpha$-randomness used in~\cite{Greenberg2018}. (For the definition of weak and strong $\alpha$-randomness and for the complexity criteria for them see~\cite[Section 13.5]{DowneyHirschfeldt}.)

\subsection*{Acknowledgements}

Authors are grateful to the participants and organizers of the Heidelberg Kolmogorov complexity program where the question of the complexity increase was raised, and to all colleagues (from the ESCAPE team, LIRMM, Montpellier, Kolmogorov seminar and HSE Theoretical Computer Science Group and other places) who participated in the discussions, in particular to Bruno Bauwens, Noam Greenberg, Konstantin Makarychev, Yury Makarychev, Joseph Miller, Alexey Milovanov, Ilya Razenshteyn, Andrei Romashchenko,  Nikolai Vereshchagin,  Linda Brown Westrick.

Special thanks to Fedor Nazarov who kindly allowed us to include his argument (using Fourier series), and, last but not least, to Peter G\'acs who explained us how the tools from~\cite{Ahlswede1976} can be used to provide the desired result about Kolmogorov complexity.

\section*{Appendix A: Simpler arguments and weaker bounds}\label{sec:weaker}

If we are interested only in \emph{some} increase of entropy and do not insist on the optimal lower bound, some simpler arguments (that do not involve entropy arguments and just prove the combinatorial statement with a weaker bound) are enough. In this section we provide two arguments of this type; the corresponding regions of parameters are shown in Figure~\ref{pic:comparison} (together with the optimal bound of Theorem~\ref{thm:main}).

\begin{figure}[h]
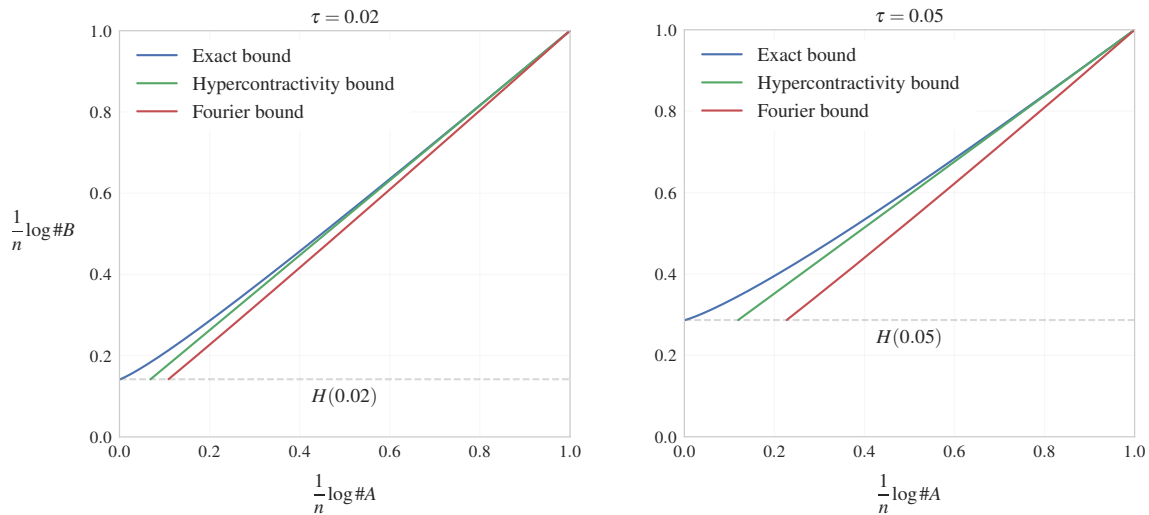

	\begin{subfigure}[b]{0.538\textwidth}
		\centering
		\resizebox{.95\linewidth}{!}{\input{"Bounds_for_tau=0.02.pgf"}}
	\end{subfigure}%
	\begin{subfigure}[b]{0.462\textwidth}
		\centering
		\resizebox{.95\linewidth}{!}{\input{"Bounds_for_tau=0.05.pgf"}}
	\end{subfigure}
	\caption{Bounds that could be obtained by different techniques}\label{pic:comparison}
\end{figure}

\subsection*{Using Fourier analysis}\label{subsec:fourier}

We start with a proof (suggested by Fedor Nazarov\footnote{see \scriptsize \url{http://mathoverflow.net/questions/247193/union-of-almost-hamming-balls}}) of a weak version of Proposition~\ref{prop:strong-combinatorial} showing that for every $\tau$ and every $\beta<1$ there exists some $\alpha<\beta$ such that required bound $\#A \le 2^{\alpha n}$ is valid for every $B$ of size $2^{\beta n}$. 

Every real-valued function on the Boolean hypercube $\B{n}$, identified with $\{-1,1\}^n$ and considered as a multiplicative group in this section, can be written in the standard Fourier basis: 
\[
f(x) = \sum_{S\subset \{1,\ldots,n\}} \widehat{f}_S \chi_{S}(x),
\]
where $\widehat{f}_S$ are Fourier coefficients, and $\chi_{S}(x) = \prod_{i\in S} x_i$. Functions $\chi_{S}$ are characters of the Boolean cube as a multiplicative group. They form an orthonormal basis in the space of real-valued functions on $\B{n}$ with respect to the following inner product:
\[ \langle f, g \rangle = \frac{1}{2^n}\sum_{x\in \B{n}} f(x)g(x) = \E_{x \in \B{n}} f(x)g(x) \]
This Fourier representation will be useful for us, since the representation of the convolution of two functions is the point-wise product of their representations: $\widehat{f*g}_S = \widehat{f}_S\, \widehat{g}_S$, where the convolution is defined as \[(f*g)(x) = \E_{t\in\B{n}} f(xt)g(t^{-1})\]
(in fact, in our case $t^{-1}=t$).

For a set $B\subset \B{n}$ we are interested in the probability 
\[
   \Noise{\tau}^{B}(x) = \Pr[N_{\tau}(x)\in B].
\]
This function is a convolution of the indicator function $\ind{B}$ of the set $B$ (equal to $1$ inside the set and $0$ outside) and the distribution of the noise, multiplied by $2^n$ (since we divide by $2^n$ when computing the expectation): 
\[\Noise{\tau}^{B} = \ind{B} * f,\]
where $f(x)=2^n \Pr[N_\tau(\ones)=x]$. Here $\ones\in\B{n}$ is the unit of the group, i.e., $\ones=(1,1,\ldots,1)$. The Fourier coefficient $\widehat{f}_S$ is easy to compute:
\[
\widehat{f}_S = \langle f, \chi_S \rangle = \E_{x\in\B{n}} f(x)\chi_S(x), 
\]
 and both functions $f$ and $\chi_S$ are products of functions depending on one coordinate:
\[
  f(x_1,\dotsc,x_n)=g(x_1) \dotsm g(x_n)
\]
where $g(1)=2(1-\tau)$ and $g(-1)=2\tau$,  and
\[
  \chi_S(x_1,\dotsc,x_n)=\chi_1(x_1) \dotsm \chi_n(x_n),
\]
where $\chi_i$ is constant $1$ if $i\notin S$, and $\chi_i(x)=x$ for $i\in S$. Due to independence, the expectation of the product is a product of expectations; they are $1$ for $i\notin S$ and $1-2\tau$ for $i\in S$, so
\[
  \widehat{f}_S=(1-2\tau)^{\#S}
\]
In other terms, noise (convolution with $f$) decreases the $S$-th coefficient of the Fourier transform by multiplying it by $(1-2\tau)^{\#S}$. We need to apply noise to the indicator function of $B$ that we denote by $b = \ind{B}$, and get a bound for the number of points where $b * f$ exceeds~$1/n$.

 Why $b*f$ cannot be relatively large (greater than $1/n$) on a large set $A$? We know that
\[
(b*f)(x)=\sum_S (1-2\tau)^{\#S}\,\widehat{\hbox{$b$}}_S \chi_S(x).
\]
This sum consists of $2^n$ terms (its elements form a vector of length $2^n$) and can be split into two parts: for ``small'' $S$, where $\#S< d$, and for ``large'' $S$, where $\#S\ge d$. Here $d$ is some threshold to be chosen later in such a way that the first part (for small $S$) does not exceed, say $1/2n$ for all $x$. Then the second part should exceed $1/2n$ everywhere on $A$, and this makes the $L_2$-norm of the second part (as a vector of the corresponding coefficients) large, while all coefficients in the second part are multiplied by small factors not exceeding $(1-2\tau)^d$.

How should we choose the threshold $d$? The coefficient $\widehat{\hbox{$b$}}_\varnothing$ equals $\mu(B)$, the uniform measure of $B$, and for all other coefficients we have $|\widehat{\hbox{$b$}}_S|\le\mu(B)$. The size (the number of terms) in the first part is the number of sets of cardinality less than $d$, and is bounded by $\poly(n)2^{n\HH(d/n)}$. Therefore, if we choose $d$ in such a way that
\[
 \mu(B)\poly(n)2^{n\HH(d/n)} \le \frac{1}{2n},
\]
we achieve our goal; the first part of the sum never exceeds $1/(2n)$.

Now the second part: compared to the same part of the sum for $b(x)$, we have all coefficients multiplied by 
$(1-2\tau)^d$ or smaller coefficients, so the $L_2$-norm of this part is bounded: 
\[
 \|\text{second part}\|_{2}\le (1-2\tau)^d\|b\|_2=(1-2\tau)^d \sqrt{\mu(B)}.
\]
On the other hand, if the second part exceeds $1/(2n)$ inside $A$, we have the lower bound: 
\[
 \|\text{second part}\|_{2}\ge \sqrt{\mu(A)}/(2n).
\]
In this way we get
\[
\sqrt{\mu(A)}/(2n) \le (1-2\tau)^d \sqrt{\mu(B)},
\]
or
\[
\mu(A) \le 4n^2 (1-2\tau)^{2d}\mu(B)
\]
where $d$ is chosen in such a way that 
\[
\mu(B)\le 2^{-n\HH(d/n)}/\poly(n)
\]
For $\#B = 2^{\beta n}$ we have $\HH(d/n)\approx 1-\beta$ and 
\[
\#A \le (1-2\tau)^{2d} 2^{\beta n} 
\]
We see that the first term gives an exponentially small factor since $d$ is proportional to $n$:
\[
d/n \approx \HH^{-1}(1-\beta)
\]
(here $\HH^{-1}(\gamma)$ is the preimage of $\gamma$ between $0$ and $1/2$). So we get the required bound for some $\alpha<\beta$ as promised.

\subsection*{Using hypercontractivity}

We can get a better bound using two-functions hypercontractivity inequality for uniform bits, whose proof can be found in \cite[chapter~10]{ODonnell2014}:
\begin{proposition}[Two-function hypercontractivity inequality]\label{prop:two_f_hypercontractivity}
	Let $f, g: \B{n} \to \mathbb{R}$, let $r, s \ge 0$, and assume $0\le 1-2\tau \le \sqrt{rs} \le 1$. Then
	\[
	\E_{\substack{x\in\B{n} \\ y=\Noise{\tau}(x)}}\!\! [f(x)g(y)] \le \|f\|_{1+r}\, \|g\|_{1+s}
	\]
\end{proposition}
Here the distribution of $x$ is the uniform distribution in $\B{n}$, and $y$ is obtained from $x$ by applying $\tau$-noise: $y = \Noise{\tau}(x)$. The same distribution can be obtained in a symmetric way, starting from $y$. The notation $\|\cdot\|_p$ denotes $L_p$-norm: 
\[
\| u \|_p = \left( \mathbb{E} |u^p| \right)^{1/p}.
\]

How do we apply this inequality? For an arbitrary set $B$ we consider the set 
     \[A = \{x: \Pr[\Noise{\tau}(x) \in B] \ge \eps \}.\]
Let $a, b$ be the indicator functions of $A$ and $B$. Then Proposition~\ref{prop:two_f_hypercontractivity} gives
 \[
\E[a(x)b(y)] = \Pr[x\in A, y\in B] \ge \Pr[x\in A] \Pr[y\in B | x\in A] \ge \mu(A) \eps.
\]
Now we write down the hypercontractivity inequality (note that $\|\ind{X}\|_{q} = \mu(X)^{1/q}$):  
\begin{align*}
\eps \mu(A) &\le \mu(A)^{1/(1+r)} \mu(B)^{1/(1+s)} \\
\log\eps + \log\mu(A) &\le \frac{\log\mu(A)}{1+r} + \frac{\log\mu(B)}{1+s} \\
\log\mu(A) &\le \frac{1+r}{r(1+s)}\log\mu(B)-\frac{1+r}{r}\log\eps.
\end{align*}
This is true for every $r,s$ with $\sqrt{rs}\ge 1-2\tau$. To get the strongest bound we minimize the right hand side, so we use (for given $r$) the minimal possible value of $s = (1-2\tau)^2/r$:
\[ \log\mu(A) \le \frac{1+r}{r+(1-2\tau)^2}\log\mu(B) - \frac{1+r}{r}\log\eps. \]
If $\eps = 1/\poly(n)$, we can set $r\to 0$ at the appropriate rate (say, $r=1/\log n$), so that the last term is still $\littleO(n)$, and we finally get:
\begin{equation*}
\begin{gathered}
\log \mu(A) \le \frac{1}{(1-2\tau)^2}\log\mu(B) +\littleO(n)\\
\log \# A   \le -\left( (1-2\tau)^{-2}-1 \right) n+ (1-2\tau)^{-2} \log \# B + \littleO(n)
\end{gathered}
\end{equation*}

\section*{Appendix B. Entropy statement and its proof}

For the reader's convenience we reproduce here the proof of Proposition~\ref{prop:noise-entropy} (following Wyner and Ziv~\cite{Wyner1973}). Let us recall what it says.

\begin{quote}
Let $P$ be an arbitrary random variable with values in $\mathbb{B}^n$, and let $P'$ be its noisy version obtained by applying $N_\tau$ independently to each bit in $P$. Choose $p\le 1/2$ in such a way that $\HH(P)=n\HH(p)$. Then consider $q=N(p,\tau)$, the probability to get $1$ if we apply $N_\tau$ to a variable that equals $1$ with probability $p$. Then $\HH(P')\ge n\HH(q)$.
\end{quote}

As we have mentioned, we start with the bound for $n=1$ and then extend it to all $n$ (``one-letter characterization'', ``tensorization'').

Let us consider a more general setting. Let $X$ and $Y$ be finite sets. Consider some stochastic transformation $T\colon X\to Y$: for every $x\in X$ we have some distribution $T(x)$ on $Y$. Then, for every random variable $P$ with values in $X$, we may consider the random variable $T(P)$ with values in $Y$. (In other words, we consider a random variable with values in $X\times Y$ whose marginal distribution on $X$ is $P$ and conditional distribution $Y\cnd X$ is $T$.) For a fixed $T$ (our main example is adding noise) we are interested in the relation between the entropies of $P$ and $T(P)$ for arbitrary $P$. In other words, we consider the set of all pairs $(\HH(P),\HH(T(P)))$ for all possible $X$-valued random variables $P$. It is a subset of the rectangle $[0,\log\#X]\times[0,\log\#Y]$. We denote this set by $S(T)$. The following lemma shows that for a product of two independent transformations $T_1\colon X_1\to Y_1$ and $T_2\colon X_2\to Y_2$ this set can be bounded in terms of the correspoding sets for $T_1$ and $T_2$.

\begin{lemma}\label{lem:tensor}
Let $T_1\colon X_1\to Y_1$ and $T_2\colon X_2\to Y_2$ be two stochastic transformations, and let $T_1\times T_2\colon X_1\times X_2\to Y_1\times Y_2$ be their product \textup(independent transformation of both coordinates\textup). Then every point $(u,u')$ in $S(T_1\times T_2)$ is above a sum of some point in $S(T_1)$ and some convex combination of points in $S(T_2)$.
\end{lemma}

Here ``above'' means ``can be obtained by increasing the second coordinate'', and a convex combination is a linear combination with non-negative coefficients that have sum $1$.

\begin{proof}
Consider some random variable $(P_1,P_2)$ with values in $X_1\times X_2$; the components $P_1$ and $P_2$ can be dependent. Then 
\[
\HH(P_1,P_2)=\HH(P_1)+\HH(P_2\cnd P_1).
\]
This is the first coordinate of a pair in question; the second coordinate is the entropy of the variable
\(
(T_1\times T_2)(P_1,P_2);
\)
its components $Q_1$ and $Q_2$ are dependent and have (marginal) distributions $T_1(P_1)$ and $T_2(P_2)$. The second coordinate of the pair is then
\[
\HH(Q_1,Q_2)=\HH(Q_1)+\HH(Q_2\cnd Q_1).
\]
We may consider all four variables $P_1,P_2,Q_1,Q_2$ as defined on the same space that is a product of three spaces: the space where $(P_1,P_2)$ is defined, the space used in the stochastic transformation of $P_1$ and the space used in the stochastic transformation of $P_2$. Now we see that the pair we are interested in is a sum of two pairs:
\[
(\HH(P_1,P_2),\HH(Q_1,Q_2))=(\HH(P_1),\HH(Q_1))+(\HH(P_2\cnd P_1),\HH(Q_2\cnd Q_1)).
\]
The first pair $(\HH(P_1),\HH(Q_1))$ is in $S(T_1)$ by definition. The second pair, as we will show, is above $(\HH(P_2\cnd P_1), \HH(Q_2\cnd P_1))$. After that we note that, by definition, the conditional entropy with condition $P_1$ is a convex combination of conditional entropies with conditions $P_1=x$ for all $x\in X_1$, and all pairs $(\HH(P_2\cnd P_1=x),\HH(Q_2\cnd P_1=x))$ are in $S(T_2)$, since for every $x$ the distribution $(Q_2\cnd P_1=x)$ is obtained by applying $T_2$ to the distribution $(P_2\cnd P_1=x)$.

It remains to show that 
\[
\HH(Q_2\cnd Q_1) \ge \HH(Q_2\cnd P_1).
\]
This is true because $Q_1$ and $Q_2$ are independent given $P_1$: the difference $\HH(Q_2\cnd Q_1)-\HH(Q_2\cnd P_1)$ is equal to $I(Q_2\,{:}\,P_1\cnd Q_1)-I(Q_1\,{:}\,Q_2\cnd P_1)$, and the second term is zero due to the conditional independence.
\end{proof}

This lemma obviously generalizes for the product of several stochastic transformations. For the noise case in $\B{n}$ we consider a product of $n$ copies of ``one-letter'' transformation $N_\tau$ that maps $0$ to $1$ with probability $\tau$ and vice versa.\footnote{We use the same notation $N_\tau$ for the one-bit transformation that we used before for applying noise to all bits of some $n$-bit string, since the meaning is clear from the context.}

\begin{lemma}\label{lem:convex}
The set $S(N_\tau)$ for the one-letter transformation $N_\tau$ is a curve in the unit square that starts at $(0,\HH(\tau))$ and ends at $(1,1)$. This curve is increasing and convex.
\end{lemma}

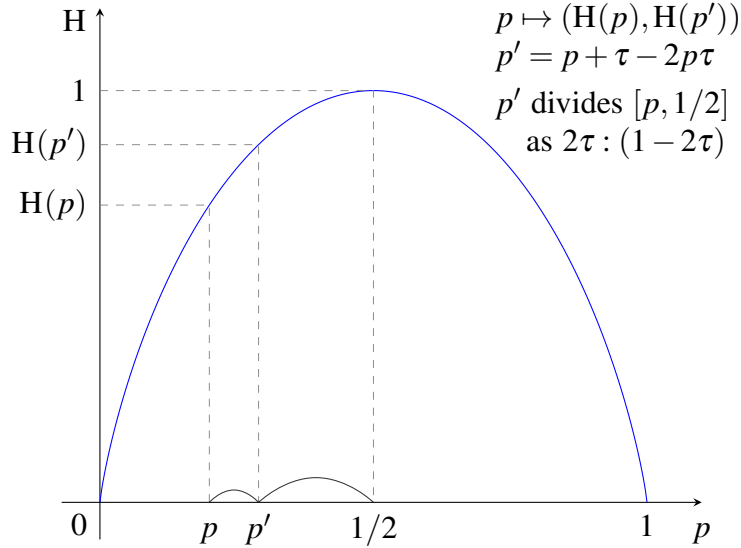
\begin{figure}
	\begin{center}
		\begin{tikzpicture}
		\begin{axis}[
		width=10cm,
		axis line style={black, opacity=1, ->},
		axis lines=middle,
		xmin=-0.07, xmax=1.1,
		ymin=-0.09, ymax=1.2,
		enlarge x limits=false,
		enlarge y limits=false,
		clip=false,
		ytick=\empty,
		xtick=\empty,
		] 
		\addplot[blue,smooth] expression[domain=0:1,samples=300] {-x*log2(x)-(1-x)*log2(1-x)};
		\node[below left] at (axis cs:0,0) {$0$};
		\draw [dashed, opacity=0.9,help lines] (axis cs:0,1) -| (axis cs:0.5,0);
		\node[below] at (axis cs:0.5,-0.003) {$1/2$};
		\node[below] at (axis cs:1,-0.003) {$1$};
		\node[left] at (axis cs:0,1) {$1$};
		\draw [dashed, opacity=0.9,help lines] (axis cs:0,0.7219280948873623) -| (axis cs:0.2,0);
		\node[above] at (axis cs:0.2,-0.127) {$p$};
		\node[left] at (axis cs:0,0.7219280948873623) {$\HH(p)$};
		\draw [dashed, opacity=0.9,help lines] (axis cs:0,0.8687212463394045) -| (axis cs:0.29,0);
		\node[above] at (axis cs:0.29,-0.1295) {$p'$};
		\node[left] at (axis cs:0,0.8687212463394045) {$\HH(p')$};
		\node[above] at (axis cs:1.1,-0.127) {$p$};
		\node[left] at (axis cs:0,1.17) {$\HH$};
		\addplot[black,tension=1,opacity=0.8,smooth] coordinates {(0.2,0) (0.245,0.03) (0.29,0)};
		\addplot[black,tension=1,opacity=0.8,smooth] coordinates {(0.29,0) (0.395,0.06) (0.5,0)};
		\node[right, align=left, fill=white] at (axis cs:0.7,1.13) {$p \mapsto (\HH(p), \HH(p'))$ \\ $p' = p+\tau-2p\tau$};
		\node[right, align=left] at (axis cs:0.7,0.92) {$p'$ divides $[p, 1/2]$ \\ \quad as $2\tau:(1-2\tau)$};
		\end{axis}
		\end{tikzpicture}
	\end{center}
	\caption{Obtaining a point $(\HH(p),\HH(p'))$ on the curve.}
\end{figure}

This curve is shown (for six specific values of $\tau$) in Figure~\ref{pic:boundary}.

\begin{proof}[Proof of Lemma~\ref{lem:convex}]
This is an exercise in elementary calculus; still we provide the sketch of a proof. The curve in question is the image of the mapping 
\[
p \mapsto (\HH(p),\HH(p')),
\]
where $p'=N(p,\tau)=p+\tau-2p\tau$, the probability to get $1$ if we choose $1$ with probability $p$ and then change the result with probability $\tau$ (independently). The point $p'$ divides the interval $[p,1/2]$ as $\,2\tau: (1-2\tau)$. Fix $\tau$, and let $p$ increase with constant speed from $0$ to $1/2$. Then $p'$ also increases with constant speed from $\tau$ to $1/2$, and the point $(\HH(p),\HH(p'))$ moves from left to right starting at $(0,\HH(\tau))$ and finishing at $(1,1)$ (when $p=1/2$, we have $p'=1/2$). Then the point goes back along the same curve, so we consider only $p\in(0,1/2)$. To show that the curve is convex, we need to check that its slope increases from left to right (as $p$ increases). Both points $p$ and $p'$ move with constant speeds, so the slope is proportional to the ratio $\HH'(p')/\HH'(p)$, where $\HH(p)=-p\log p -(1-p)\log(1-p)$. To compute the derivative $\HH'(p)$, we may replace the binary logarithms by natural ones (this does not change the ratio of derivatives). The derivative of $p\ln p$ is $\ln p+1$, so
\[
\HH'(p)= -\ln p - 1 + \ln (1-p) + 1 = \ln\left(\frac{1-p}{p}\right).
\]
For computations, it is convenient to shift the origin and let $p=\frac{1}{2}+u$, then 
\[
\HH'\left(\tfrac{1}{2}+u\right)=\ln\left(\frac{1-2u}{1+2u}\right).
\]
In this new coordinates $p'$ corresponds to $u'$ that is proportional to $u$, i.e., $u'=cu$, where $c$ is a constant ($c=1-2\tau$). We need to show that $\HH'\left(\frac{1}{2}+u'\right)/\HH'\left(\frac{1}{2}+u\right)$ increases as $u$ increases from $-1/2$ to $0$. Letting $u=-v/2$, we need to show that 
\[
\ln\left(\frac{1+cv}{1-cv}\right)/\ln\left(\frac{1+v}{1-v}\right)
\]
increases as $v$ decreases from $1$ to $0$. Using the series
\[
\ln\left(\frac{1+v}{1-v}\right)=\ln(1+v)-\ln(1-v)=2(v+\tfrac{1}{3}v^3+\tfrac{1}{5}v^5+\ldots),
\]
we can reformulate our statement as follows: the ratio
\[
\frac{
  c\cdot v+c^3 \cdot \tfrac{1}{3}v^3+c^5 \cdot \tfrac{1}{5}v^5+\ldots
}{
  v+ \tfrac{1}{3}v^3+\tfrac{1}{5}v^5+\ldots
}
\]
decreases as $v$ increases from $0$ to $1$. This ratio is a center of gravity for points having (decreasing) coordinates $c,c^3,c^5,\ldots$ and masses $v, \frac{1}{3}v^3, \frac{1}{5}v^5,\ldots$. When $v$ is small, the first mass ($v$) is the most important (others are much smaller); as $v$ increases, the other masses become more and more important.  The sequence of coordinates decreases, so the center of gravity moves to the left as required. To say it a bit more formally, we note that the ratio of the first mass ($v$) and the rest ($\frac{1}{3}v^3+\frac{1}{5}v^5+\ldots$) decreases as $v$ increases, so the center of gravity become closer to the center of gravity for the system without the first mass, and the latter also moves to the right for similar reasons. To make the formal inductive proof, we need to prove the similar statement for finitely many masses and then consider the limit.\footnote{In general, we use the following monotonicity statement: \emph{if the coordinates of points are $x_1>x_2>\ldots>x_n$ and the masses $m_1,\ldots,m_n$ are changed in such a way that new masses $m_i'$ satisfy the inequality $m_j'/m_i'>m_j/m_i$ for $j>i$, then the center of gravity moves to the left after the change}. This can be easily proven by induction over $n$, following the scheme explained above.}
\end{proof}

Lemma~\ref{lem:convex} shows that the set of points of the unit square above $S(N_\tau)$ is convex. Therefore, applying Lemma~\ref{lem:tensor} for the noise case, we do not need convex combinations: one point in each set $S(T_i)$ is enough. Note also that for $N$ copies we have a sum of $N$ points above $S(N_\tau)$, and dividing this sum by $N$, we get a point in $S(N_\tau)$, as required. Proposition~\ref{prop:noise-entropy} is proven.

In other words, for a fixed entropy $\HH(P)$ the minimal entropy of $P'$ is achieved for the Bernulli distribution $P=B_p$ for a suitably chosen $p$. As we have said, this is the Shannon information theory version of our main result about increasing complexity by random noise. 

\section*{Appendix C: The proof of McDiarmid's inequality}

In this section we reproduce the standard proof of McDiarmid's inequality for the reader's convenience. We start with a technical lemma about the expectation of an exponent of a bounded random variable.

\begin{proposition}[H\"oeffding's lemma]
	If any two values of a real random variable $U$ differ at most by $c$, then $\E\exp(U-\E U)\le \exp(c^2/8)$.
\end{proposition}

\begin{proof}
In this statement we may change $U$ by a constant, so we may assume that $0\le U\le c$. Then $\E U$ is equal to $pc$ for some $p\in [0,1]$. The exponent function is convex, therefore
\[
\exp(u) \le \frac{c-u}{c}\cdot 1 + \frac{u}{c}\cdot e^c
\]
for $u\in[0,c]$, and 
\[
 \E \exp (U) \le \frac{c-\E u}{c} + \frac{\E U}{c}\cdot e^c=1-p+pe^c.
\]
We need to show that
\[
\E \exp (U-\E U)=\E \exp(U)/\exp(\E U)\le (1-p+pe^c)/e^{pc} \le \exp(c^2/8).
\]
Taking logarithms, we need to show that for all $t\ge 0$ and $p\in[0,1]$ we always have
\[
\varphi(t):= \ln (1-p+pe^t) - pt \le t^2/8.
\]
Note that we replaced $c$ by $t$ since we plan to consider the left hand side as a function of $t$ (for fixed $p$) and compute its derivatives (and $c$ looks more like a notation for a constant). To prove this inequality for all $t\ge0$, it is enough to show that
\[
\varphi(0)=0,\ \varphi'(0)=0,\ \text{and }\ \varphi''(t)\le 1/4\ \text{for all $t\ge 0$} 
\]
(use Taylor's formula or just integrate twice: the first integration gives $\varphi'(t)\le t/4$ for $t\ge0$). The equality $\varphi(0)=0$ is obvious; for the two other claims we have to compute
\[
\varphi'(t)=\frac{1}{1-p+pe^t}\cdot pe^t - p,\quad
\varphi''(t)=-\frac{1}{(1-p+pe^t)^2}(pe^t)^2+ \frac{1}{1-p+pe^t}\cdot pe^t
\]
We see immediately that $\varphi'(0)=0$; for the second inequality we have to rewrite
\begin{multline*}
\varphi''(t)=-\frac{1}{(1-p+pe^t)^2}(pe^t)^2+ \frac{1}{1-p+pe^t}\cdot pe^t=\\=
\frac{-(pe^t)^2+(1-p+pe^t)pe^t}{(1-p+pe^t)^2}=
\frac{(1-p)pe^t}{(1-p+pe^t)^2}.
\end{multline*}
The last expression has the form $uv/(u+v)^2$ for $u=1-p$ and $v=pe^t$, and therefore does not exceed $1/4$.
\end{proof}

This lemma is a key step in the proof of an inequality about martingales, the \emph{Azuma--H\"oeffding inequality}. Consider finite probability spaces $X_1,\ldots,X_n$ and the product probability space $X_1\times\ldots\times X_n$ (with independent coordinates). Let $U_0,\ldots,U_n$ be a sequence of random variables defined on the product space. We assume that $U_i(x_1,\ldots,x_n)$ depends only on the first $i$ arguments. In particular, $U_0$ is a constant, and $U_i$ can be written as $U(x_1,\ldots,x_i)$. Assume that $U_0,\ldots,U_n$ is a martingale, i.e., the expected value $\E_{x_i\in X_i}U_i(x_1,\ldots,x_i)$ for fixed values of $x_1,\ldots,x_{i-1}$ equals $U_{i-1}(x_1,\ldots,x_{i-1})$ for every $x_1\in X_1,\ldots,x_{i-1}\in X_{i-1}$.

\begin{proposition}[Azuma-H\"oeffding inequality]
Assume that for some constants $c_1,\ldots,c_n$ and for all $i=1,\ldots, n$ the following condition is satisfied: $U_i(x_1,\ldots,x_i)$ changes at most by $c_i$ if we change $x_i$ leaving the other arguments $x_1,\ldots,x_{i-1}$ unchanged. 
Then the following inequality holds:
	\[ \Pr[U_n-U_0\ge z] \le \exp\Big(-\frac{2z^2}{\sum_{i=1}^n c_i^2}\Big). \]
\end{proposition}

\begin{proof}[Proof of the Azuma-H\"oeffding inequality]
For arbitrary positive $t$ we may write the Markov inequality for the random variable $\exp (t(U_n-U_0))$:
\[
	\Pr[U_n-U_0 \ge z] = \Pr[\exp(t(U_n-U_0)) \ge \exp(tz)] \le \exp(-tz)\E\exp(t(U_n - U_0)).
\]
The expectation $\E\exp(t(U_n - U_0))$ can be rewritten as 
\begin{multline*}
\E_{x_1,\ldots,x_n}\exp(t(U_n - U_0))=\E_{x_1,\ldots,x_{n}} \left(\exp(t(U_{n-1}-U_0))\exp (t(U_n-U_{n-1}))\right) =\\
= \E_{x_1,\ldots,x_{n-1}} \exp(t(U_{n-1}-U_0)) \E_{x_n|x_1,\ldots,x_{n-1}} \exp (t(U_n-U_{n-1})).
\end{multline*}
For every fixed $x_1,\ldots,x_{n-1}$ we can apply H\"oeffding's lemma to the internal expectation: it is bounded by $\exp(t^2c_n^2/8)$. The factor $\exp(t^2c_n^2/8)$ is a constant, so we may take out this factor and continue:
\[
\ldots \le   \exp(t^2c_n^2/8)\E_{x_1,\ldots,x_{n-1}} \exp(t(U_{n-1}-U_0)).
\]
Then the same procedure is repeated with $x_{n-1},\ldots,x_1$:
\[
\ldots \le \exp(t^2c_n^2/8)\exp(t^2c_{n-1}^2/8)\E_{x_1,\ldots,x_{n-2}} \exp(t(U_{n-2}-U_0))
\le\ldots\le  \exp\left(\frac{t^2\sum_{i=1}^n c_i^2}{8}\right).
\]
We conclude that
\[
\Pr[U_n-U_0 \ge z] \le \exp(-tz)\exp\left(\frac{t^2\sum_{i=1}^n c_i^2}{8}\right).
\]
This is true for every $t\ge0$, so we chose the value of $t$ that makes the right-hand side (exponent of a quadratic polynomial) minimal and get the required inequality.
\end{proof}
Now we get the McDiarmid inequality (Proposition~\ref{prop:McDiarmid}) as an easy consequence of the Azuma--H\"oeffding inequality.
\begin{proof}[Proof of the McDiarmid inequality]
The McDiarmid inequality deals with a function $f$ on $X_1\times\ldots\times X_n$ that changes at most by $c_i$ when $i$th argument is changed. We apply the Azuma-H\"oeffding inequality to functions $U_i(x_1,\dots,x_i)$ that are expectations of $f$ when $x_1,\ldots,x_i$ are fixed:
\[
U_i(x_1,\ldots,x_i)= \E_{x_{i+1},\ldots,x_n}[f(x_1,\ldots,x_n)\mid x_1, \ldots, x_i].
\]
When we fix $x_1,\ldots,x_{i-1}$ and change $x_i$, the expectation that defines $U_i$ changes at most by~$c_i$. Indeed, for every fixed values of $x_{i+1},\ldots,x_n$ the value of $f(x_1,\ldots,x_n)$ changes at most by $c_i$ when changing $x_i$ (due to our assumption about $f$). So the same is true for $U_i$ (and $U_i-U_{i-1}$ as well), and we may apply the Azuma--H\"oeffding inequality. In this way we get the required bound.
\end{proof}

\bibliographystyle{plain}
\bibliography{bibliography}

\end{document}